\documentclass{article} 
\usepackage[colorlinks,linkcolor=blue,citecolor=red,urlcolor=blue]{hyperref}
\usepackage[centertags,sumlimits,intlimits,namelimits,reqno]{amsmath}
\usepackage{latexsym,amsfonts,amssymb,exscale,enumerate,amsthm}
\usepackage[all]{xy}

        \newcommand\N{{\mathbb N}}   
           
        \newcommand\R{{\mathbb R}}   
        \newcommand\maps{{\colon}}   
        \newcommand{\define}[1]{{\bf \boldmath #1}}

\newtheorem{theorem}{Theorem}
\newtheorem{definition}[theorem]{Definition}
\newtheorem{lemma}[theorem]{Lemma}

\theoremstyle{remark}

\newtheorem*{example*}{Example}
       
\begin{document}   	
	
	\begin{center}   
	{\bf  Quantum Techniques for Reaction Networks \\}   
        \vspace{0.3cm}
	{\em John\ C.\ Baez \\}
        \vspace{0.3cm}
        {\small
        Department of Mathematics, University of California, \\ 
        Riverside CA 92521, USA \\
         and \\
        Centre for Quantum Technologies,
        National University of Singapore, \\
         Singapore 117543}
        \vspace{0.3cm}
        {\small email:  baez@math.ucr.edu\\} 
	\vspace{0.3cm}   
	{\small June 14, 2013}
	\vspace{0.3cm}   
	\end{center}   

\begin{abstract}
\noindent
Reaction networks are a general formalism for describing collections
of classical entities interacting in a random way.  While reaction
networks are mainly studied by chemists, they are equivalent to Petri
nets, which are used for similar purposes in computer science and
biology.  As noted by Doi and others, techniques from quantum field
theory can be adapted to apply to such systems.  Here we use these
techniques to study how the `master equation' describing stochastic
time evolution for a reaction network is related to the `rate
equation' describing the deterministic evolution of the expected
number of particles of each species in the large-number limit.  We
show that the relation is especially strong when a solution of master
equation is a `coherent state', meaning that the numbers of entities
of each kind are described by independent Poisson distributions.
\end{abstract}

\section{Introduction}

A `reaction network' describes how various kinds of classical
particles can interact and turn into other kinds.  They are commonly
used in chemistry.  For example, this reaction network gives a very
simplified picture of what is happening in a glass of water:
\[
\xymatrix{
\mathrm{H}_2\mathrm{O} \ar@<0.6ex>[r] 
& \mathrm{H}^+ + \mathrm{OH}^-   \ar@<0.6ex>[l] 
& \mathrm{H}_2\mathrm{O} + \mathrm{H}^+ \ar@<0.6ex>[r] 
& \mathrm{H}_3\mathrm{O}^+ \ar@<0.6ex>[l]
}
\]
The reactions here are all reversible, but this is not required by the
reaction network formalism.  The relevant particles here are not
atoms, but rather `species' including the water molecule
$\mathrm{H}_2\mathrm{O}$, the proton $\mathrm{H}^+$, the hydroxyl ion
$\mathrm{OH}^- $ and the hydronium ion $\mathrm{H}_3\mathrm{O}^+$.

While the dynamics of the reactions is fundamentally
quantum-mechanical, reaction networks give a simplified description of
what is going on.  First, given a `rate constant' for each reaction,
one can write down a differential equation called the \emph{rate
equation}, which describes the time evolution of the concentration of
each species in a deterministic way.  This is especially useful in the
limit where there are many particles of each species.  At a more
fine-grained level, one can describe reactions stochastically,
using the \emph{master equation}.  In the rate equation the
concentration of each species is treated as a continuous variable, but
in the master equation, the number of particles of each species is an
integer.

In this paper we explain how techniques from quantum field theory can
be used to study reaction networks.  This work is part of a broader
program emphasizing the links between quantum mechanics and a subject
we call `stochastic mechanics', where probabilities replace amplitudes
\cite{BaezBiamonte:2012,BaezFong:2012,BaezFong:2014}.

Using ideas from quantum field theory to model classical stochastic
systems goes back at least to 1976, with the work of Doi \cite{Doi:1976}.  It
has been further developed by Grassberger, Scheunert
\cite{GrassbergerScheunert:1980} and
many others.  The big surprise one meets at the very beginning of this
subject is that the canonical commutation relations between
creation and annihilation operators, long viewed as a hallmark of
quantum theory, are also perfectly suited to systems of identical
\emph{classical} particles interacting in a probabilistic way.

In quantum theory, the simplest example of a Fock space is used to describe 
the states of a quantum harmonic oscillator.    When such an oscillator is in its
$\ell$th energy level, its state is a vector that we can formally write as $z^\ell$ for 
some formal variable $z$---or better, $z^\ell/\sqrt{\ell!}$, where the normalization
factor comes in handy later.    A general quantum state is a linear combination of 
such vectors, so we can write it as a power series 
\[
          \Psi = \sum_{\ell = 0}^\infty \psi_\ell  \frac{z^\ell}{\sqrt{\ell!}}    
\]
where the coefficient $\psi_\ell \in \mathbb{C}$ is the amplitude for the
oscillator to be in its $\ell$th energy level.   We can increase the energy level using
the `raising operator' $a^\dagger$, which multiplies any power series by $z$:
\[        a^\dagger \Psi = z \Psi  .\]
We can decrease the energy level using the `lowering operator' $a$, which differentiates
any power series:
\[       a \Psi = \frac{d}{d z} \Psi  .\]
If we choose an inner product for which the vectors $z^\ell/\sqrt{\ell!}$ are orthonormal, these operators are indeed adjoint to each other:
\[          \langle a \Psi, \Phi \rangle = \langle \Psi, a^\dagger \Phi \rangle. \]
In checking this one sees why the normalization factor is required.

We can think of $z^\ell$ as a state containing $\ell$ quanta of energy.   In applications 
to quantum field theory, we think of these quanta as actual particles---for example, photons.   If we have $k$ different kinds of particles, or a single kind of particle with 
$k$ different states, we describe the quantum state of a collection of particles 
using a power series in $k$ variables.  For example, if $k = 2$ there is a Fock space consisting of power series in two variables $z_1$ and $z_2$:
\[       \Psi = \sum_{\ell_1 = 0}^\infty \sum_{\ell_2 = 0}^\infty \psi_{\ell_1,\ell_2} 
\frac{z_1^{\ell_1}}{\sqrt{\ell_1!}} \frac{z_2^{\ell_2}}{\sqrt{\ell_2!}}  \]
where $\psi_{\ell_1,\ell_2} \in \mathbb{C}$ is the amplitude for having $\ell_1$ particles 
of the first kind and $\ell_2$ particles of the second kind.   We have operators that create 
and annihilate particles of each kind:
\[      a^\dagger_i \Psi = z_i \Psi , \quad a_i \Psi = \frac{\partial}{\partial z_i} \Psi .\]
Each creation operator $a^\dagger_i$ is adjoint to the corresponding annihilation 
operator $a_i$, and these operators obey the `canonical commutation relations':
\[ 
\begin{array}{ccl}   
  \lbrack a_i, a_j \rbrack &=& 0 \\ 
 \lbrack {a_i}^\dagger, {a_j}^\dagger \rbrack &=& 0  \\
 \lbrack a_i, a^\dagger_j \rbrack &=& \delta_{ij} .
\end{array}
\] 
These rules are the basis of many practical calculations in particle physics, where 
products of creation and annihilations operators are used to model processes in
which particles interact and turn into other particles.

Doi noticed that all this has a stochastic analogue where we work with probabilities 
rather than amplitudes.  For example, there is a stochastic Fock space whose elements
are power series
\[       \Psi = \sum_{\ell_1 = 0}^\infty \sum_{\ell_2 = 0}^\infty \psi_{\ell_1,\ell_2} 
z^{\ell_1} z^{\ell_2} . \]
Here $\psi_{\ell_1,\ell_2} \in \mathbb{R}$ is the \emph{probability} of having $\ell_1$ particles of the first kind and $\ell_2$ particles of the second kind.  The formalism is
somewhat simpler than in the quantum case: for example, the normalization factors are
no longer required.   For $\Psi$ to describe a probability distribution, we require that all 
the coefficients be nonnegative and sum to 1.  We call $\Psi$ with these properties a `mixed state'.   A special case is a monomial $z_1^{\ell_1} z_2^{\ell_2}$.   This is called a `pure state': it describes a completely definite situation where there are $\ell_1$ particles of the first kind and $\ell_2$ of the second kind.

This method of describing a probability distribution using a power series is far from new.
It was introduced by DeMoivre \cite{DeMoivre:1730}, and perfected by Laplace 
\cite{Laplace:1782} in a paper published in 1782.   He dubbed these power series `generating functions', and they have been used in probability theory ever since.    Doi's new realization was that spaces of generating functions are analogous to the Fock spaces in quantum theory.  In particular, products of creation and annihilation 
operators can be used to describe stochastic processes in which objects interact and
turn into other objects.

A natural example is chemistry, where reactions are often described
stochastically using the master equation or deterministically using the rate
equation.  There is a long line of work on both these equations, initated by Horn and 
Jackson \cite{HornJackson:1972} along with Feinberg \cite{Feinberg:1979} in the 
late 1970s, and continuing to this day. This line of work does not use quantum
techniques, but it might profit from them, and also contribute to 
their development.  To illustrate how, here we use these techniques to study how 
the master equation reduces to the rate equation in the large-number limit. This 
is formally similar to the `classical limit' in which a quantum field theory reduces to a classical field theory.  

In Theorem \ref{thm.expected.value} we start with the master equation
and derive an equation for the rate of change of the expected number
of particles of each species.  This equation is similar to the rate
equation, and reduces to it in a certain approximation.  However, in
Theorem \ref{thm.coherent} we show that this equation gives the rate 
equation \emph{exactly} when the numbers of particles of each species
are described by independent Poisson distributions.  This relies on a nice fact: 
the generating function of a product of independent Poisson distributions is an 
eigenvector of all the annhilation operators. 

In quantum mechanics, eigenvectors of the annihilation operators are called `coherent 
states'.   Many facts about classical mechanics have simple generalizations 
to quantum mechanics if we restrict attention to coherent states \cite{KlauderSkagerstam:1980}.   
Here we are seeing a similar phenomenon in stochastic mechanics: the deterministic 
dynamics given by the rate equation matches the stochastic dynamics given by
the master equation in the special case of coherent states.  We explore this 
further in a companion paper \cite{BaezFong:2014}.

The application of quantum techniques to reaction networks is just one of
many interesting connections between quantum physics and network theory.
For example, the evolution of networks displays phase transitions analogous
to Bose--Einstein condensation \cite{BianconiBarabasi:2001,JavaroneArmano:2012}.
For more, see the review by Biamonte, Faccin and De Domenico  \cite{BiamonteFaccinDeDomenico:2017}.

\section{Reaction networks}
\label{reaction.networks}

A reaction network, first defined by Aris \cite{Aris:1965}, consists of a set
of different kinds of particles---each kind being called a
`species'---together with a set of processes, called `reactions', that
turn collections of particles of various species into collections of
particles of other species.  As the name suggests, it is convenient to
draw a reaction network as a graph.

While reaction networks first arose in chemistry, their uses are not
limited to this context.  Here is an example that arose in work on 
HIV, the human immunodeficiency virus \cite{Korobeinikov:2004}:
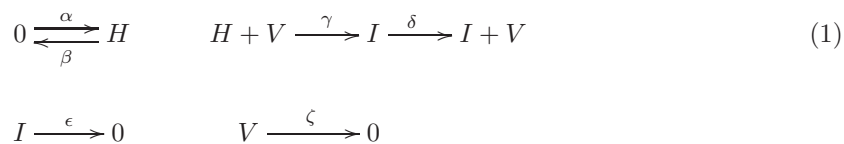
\begin{equation}
\xymatrix{
0 \ar@<0.6ex>[r]^-{\alpha} & H \ar@<0.6ex>[l]^-{\beta}  & 
H + V \ar[r]^<<<<<\gamma & I \ar[r]^>>>>>>\delta & I + V \\
 I  \ar[r]^\epsilon &0  &  V  \ar[r]^\zeta &0 
}
\label{HIV}
\end{equation}
Here we have three species:
\begin{itemize}
\item $H$: healthy white blood cells,
\item $I$: infected white blood cells, 
\item $V$: virions (that is, individual virus particles).
\end{itemize}
We also have six reactions:
\begin{itemize}
\item $\alpha$: the birth of one healthy cell, which has no input and
one $H$ as output.
\item $\beta$: the death of a healthy cell, which has one $H$ as input
and no output.
\item $\gamma$: the infection of a healthy cell, which has one $H$ and
one $V$ as input, and one $I$ as output.
\item $\delta$: `production', or the reproduction of the virus in an
infected cell, which has one $I$ as input and one $I$ and one $V$ as
output.
\item $\epsilon$: the death of an infected cell, which has one $I$ as
input and no output.
\item $\zeta$: the death of a virion, which has one $V$ as input and
no output.
\end{itemize}
So, we have a finite set of species  
\[   S = \{H, I, V \}  \]
but reactions go between `complexes', which are finite linear combinations 
of species with natural number coefficients, such as $H + V$ or $2H + 3 V$ 
or $0$.  (In the literature on reaction networks the complex $0$, corresponding to 
nothing at all, is often denoted $\emptyset$.)  We can think of complexes as 
elements of $\N^S$, the set of functions from $S$ to the natural numbers.  The 
reaction network is a graph with certain complexes as vertices and reactions as edges.

We can formalize this as follows.  There are various kinds of graphs;
the kind we want are sometimes called `directed multigraphs' or
`quivers', but to keep our terminology simple, we make the followng
definition:

\begin{definition} A \define{graph} consists of a set $E$ of \define{edges}, 
a set $V$ of \define{vertices}, and \define{source} and
\define{target} maps $s,t \maps E \to V$ saying where each edge starts
and ends.
\end{definition}

Then we may say: 

\begin{definition}
A \define{reaction network} $(S,K,T,s,t)$ consists of:
\begin{itemize}
\item a finite set $S$ of \define{species},
\item a finite set of \define{complexes} $K \subseteq \N^S$,
\item a graph $s,t \maps T \to K$ with complexes as vertices and some
finite set of \define{reactions} $T$ as edges.
\end{itemize}
\end{definition}

\noindent 
In the chemistry literature this sort of thing is called a `chemical
reaction network', but we want to emphasize that they are useful more
generally, so we drop the adjective.  Petri nets are an alternative
formalism that is entirely equivalent to reaction networks, often used
in subjects other than chemistry \cite{Koch:2010,KochReisigSchreiber:2011,Wilkinson:2006}.  
In the Petri net
literature the species are called `places' and the reactions are
called `transitions'.  For a more complete translation guide between
reaction networks and Petri nets, see \cite{BaezFong:2014}.

For convenience we shall often write $k=|S|$ for the number of species
present in a reaction network, and identify the set $S$ with the set
$\{1, \dots, k\}$.  This lets us write any complex as a $k$-tuple of
natural numbers.  In particular, we write the source and target of any
reaction $\tau$ as
\[    s(\tau) = (s_1(\tau), \dots, s_k(\tau)) \in \N^k , \]
\[    t(\tau) = (t_1(\tau), \dots,  t_k(\tau)) \in \N^k . \]

\section{The rate equation}
\label{rate.equation}

The amount of each species is represented by a `state' of the chemical
reaction network.  There are various options here. A \define{pure
state} is a vector $x \in \N^{k}$ with $i$th entry $x_i$ specifying
the number of instances of the $i$th species.  This generalises in two
ways.  First, we define a \define{classical state} to be any vector $x
\in [0,\infty)^k$ of nonnegative real numbers, and think of such a
state as specifying the expected number or concentration of the
species present.  Second, we define a \define{mixed state} to be a
probability distribution over the pure states.  This assigns a
probability $\psi_\ell$ to each $\ell \in \N^k$.  The rate equation
describes the time evolution of a classical state, while the master
equation describes the time evolution of a mixed state.

To write down either of these equations, we must first specify the
system's dynamics. In this paper we consider systems that evolve by
the law of mass action \cite{Aris:1965}. Very roughly, this law states that
the rate at which a reactions occur is proportional to the product of
the concentrations of all its input species.  We call the constant of
proportionality for each reaction its \define{rate constant}:

\begin{definition}
A \define{stochastic reaction network} $(S,K,T,s,t,r)$ consists of a
reaction network together with a map $r \maps T \to (0,\infty)$
assigning a \define{rate constant} to each reaction.
\end{definition}

\noindent 
This concept is equivalent to another concept in the literature, namely that
of a stochastic Petri net \cite{BaezFong:2014}.

Given a stochastic reaction network, the rate equation says that for each 
reaction $\tau$ the time derivative of a classical state $x$ is the product
of:
\begin{itemize}
\item the vector $t(\tau) - s(\tau)$ whose $i$th component is the change
in the number of the $i$th species due to the reaction $\tau$;
\item the concentration of each input species $i$ of $\tau$ raised to
the power given by the number of times it appears as an input, namely
$s_i(\tau)$;
\item the rate constant $r(\tau)$ of $\tau$.
\end{itemize}

\begin{definition} 
The \define{rate equation} for a stochastic reaction network 
$(S,K,T,s,t,r)$ is
\begin{equation}
\frac{d}{d t}x(t) = 
\sum_{\tau \in T} r(\tau) (t(\tau) - s(\tau)) x(t)^{s(\tau)}
\label{eq:rate_equation}
\end{equation}
where $x \maps \R \to [0,\infty)^k$ and we have
used multi-index notation to define
\[  x^{s(\tau)} = x_1^{s_1(\tau)} \cdots x_k^{s_k(\tau)}. \]
\end{definition}

For the HIV reaction network in Equation \ref{HIV}, if we also use the
Greek letter names of the reactions as names for their rate
constants, we get this rate equation:
\[ \begin{array}{ccl}
\displaystyle{ \frac{dH}{dt}} &=& \alpha - \beta H - \gamma H V
\\  \\
\displaystyle{ \frac{dI}{dt}} &=&   \gamma H V - \epsilon I
\\    \\
\displaystyle{ \frac{dV}{dt}} &=& - \gamma H V + \delta I  - \zeta V
\end{array}
\]
This example involves only complexes where the coefficient of each
species is at most 1.  It is worthwhile comparing an example with larger coefficients,
for example the dissociation of a diatomic gas
into atoms and the recombination of these atoms into a molecule:
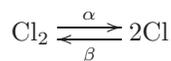
\[
\xymatrix{
\mathrm{Cl}_2 \ar@<.5ex>[r]^-{\alpha} & 2\mathrm{Cl}  \ar@<.5ex>[l]^-{\beta}
}
\]
This reaction network gives the following rate equation:
\[ \begin{array}{ccl}
\displaystyle{ \frac{d\mathrm{Cl}_2}{dt}} &=& -\alpha \mathrm{Cl}_2 + \beta \mathrm{Cl}^2
\\  \\
\displaystyle{ \frac{d \mathrm{Cl}}{dt}} &=& 2 \alpha \mathrm{Cl}_2 - 2 \beta \mathrm{Cl}^2
\end{array}
\]
The recombination of two atoms of chlorine into a diatomic molecule occurs 
at a rate proportional to the square of the concentration of atomic chlorine, 
so both terms arising from reaction $\beta$ are proportional to $\mathrm{Cl}^2$.  
The reaction $\alpha$ produces two atoms of $\mathrm{Cl}$ while the reaction 
$\beta$ uses up two atoms of $\mathrm{Cl}$, so both terms in the second 
equation have a coefficient of 2.  All this follows from Equation \ref{eq:rate_equation}.

\section{The master equation}

The master equation describes the time evolution of mixed states.
For the rest of this paper we fix a stochastic reaction
network $(S,K,T,s,t,r)$.  Recall that for each $\ell \in \N^k$, a
`mixed state' gives the probability $\psi_\ell$ that for each $i$,
exactly $\ell_i \in \N$ instances of the $i$th species are present.
The \define{master equation} says that
\[     \displaystyle{ \frac{d}{dt} } \psi_{\ell'} = 
\sum_{\ell \in \N^k} H_{\ell' \ell} \, \psi_\ell  \]
for some matrix of numbers $H$ determined by the stochastic reaction
network.  The matrix element $H_{\ell' \ell}$
is the probability per time for the pure state $\ell$ to evolve to the 
pure state $\ell'$.  This probability is a sum over reactions:
\begin{equation}
     H_{\ell' \ell} = \sum_{\tau \in T} H(\tau)_{\ell' \ell} .  
\label{def.ham.1} 
\end{equation}

To describe the matrix $H(\tau)$, we need `falling powers'.
For any natural numbers $n$ and $p$, we define the $p$th \define{falling 
power} of $n$ by 
\[     n^{\underline{p}} = n(n - 1) \cdots (n - p +1)  .\]
This is the number of ways of choosing an ordered $p$-tuple of
distinct elements from a set with $n$ elements.  Note that
$n^{\underline{p}} = 0$ if $p > n$.  More generally, for any
multi-indices $\ell$ and $m$, we define
\[  \ell^{\underline{m}} = 
\ell_1^{\;\underline{m_1}} \; \cdots \; \ell_k^{\;\underline{m_k}} \]
This is the number of ways to choose, for each species $i = 1, \dots,
k$, an ordered list of $m_i$ distinct things from a collection of
$\ell_i$ things of that species.  Thus, we expect a factor of this
sort to appear in the master equation.

Using this notation, we define the matrix $H(\tau)$ as follows:
\begin{equation}
   H(\tau)_{\ell' \ell} = r(\tau) \;  \ell^{\underline{s(\tau)}}  \; 
\left(\delta_{\ell', \ell + t(\tau) - s(\tau)} - \delta_{\ell' \ell} \right) . 
\label{def.ham.2}
\end{equation}
Here $\delta$ is the Kronecker delta, which equals 1 if its subscripts
are equal and 0 otherwise.  The first term describes the rate for the 
complex $\ell$ to become the
complex $\ell'$ via the reaction $\tau$: it equals the rate constant
for this reaction times the number of ways this reaction can occur.
The second term describes the rate at which $\ell$ goes away due to
this reaction.  Thus, the size of the second term is equal to that of
the first, but its sign is opposite.

Putting together Equations \ref{def.ham.1} and \ref{def.ham.2}, we obtain
this formula for the Hamiltonian:
\begin{equation}
H_{\ell' \ell} = \sum_{\tau \in T} r(\tau) \;  
\ell^{\underline{s(\tau)}}  \; \left(
\delta_{\ell', \ell + t(\tau) - s(\tau)} - \delta_{\ell' \ell} \right) .
\label{master1} 
\end{equation}

\section{The stochastic Fock space}

We can understand the Hamiltonian for the master equation at a deeper level 
using techniques borrowed from quantum field theory.  The first step is to 
introduce a stochastic version of Fock space.  To do this, we write any mixed 
state as a formal power series in some variables $z_1, \dots, z_k$, with the 
coefficient
of 
\[    z^\ell = z_1^{\ell_1} \cdots z_k^{\ell_k} \]
being the probability $\psi_\ell$.  That is, we write any mixed state as
\[   \Psi = \sum_{\ell \in \N^k} \psi_\ell z^\ell.\]
Because a mixed state is a probability distribution, the coefficients 
$\psi_\ell$ must be nonnegative and sum to 1. Indeed, in this formalism 
\define{mixed states} are precisely the formal power series $\Psi$ with 
\[
   \sum_{\ell \in \N^k} \psi_\ell = 1  , \qquad \psi_\ell \ge 0.
\]
The simplest mixed states are the monomials $z^\ell$; these are the
\define{pure states}, where there is a definite number of things of
each species.

In quantum mechanics we use a similar sort of Fock space, but with
power series having complex rather than real coefficients.  To obtain a 
Hilbert space, we often restrict attention to power series obeying for
which a certain norm is finite.   However, power series not obeying this 
condition are also useful, for example to ensure that operators of interest 
are everywhere defined, instead of merely densely defined.   

Thus, for the purposes of studying the master equation, we define
the \define{stochastic Fock space} to be the vector space 
$\R[[z_1, \dots, z_k]]$ of all real formal power series in the 
variables $z_1, \dots, z_k$.  We treat the Hamiltonian $H$ for 
the master equation as the operator on stochastic Fock space with
\[ H \Psi = \sum_{\ell, \ell' \in \N^k} H_{\ell' \ell} \psi_\ell 
\, z^{\ell'} \]   
for all 
\[   \Psi = \sum_{\ell \in \N^k} \psi_\ell z^\ell, \]
where the matrix entries $H_{\ell' \ell}$ are given by Equation
\ref{master1}.

This allows us to express the Hamiltonian in terms of
certain special operators on the stochastic Fock space: the
creation and annihilation operators.  Our notation here follows 
that used in quantum field theory, where the creation and annihilation 
operators are adjoints of each other.  Let $1 \le i \le k$. The 
\define{creation operator} $a_i^\dagger$ is given by
\[
a_i^\dagger \Psi = z_i \Psi.
\]
This takes any pure state to the pure state with one additional
instance of the $i$th species:
\[    a_i^\dagger \left(z_1^{\ell_1} \cdots z_i^{\ell_i}\cdots
z_k^{\ell_k}\right) = z_1^{\ell_1} \cdots z_i^{\ell_i+1} \cdots z_k^{\ell_k} .\]
The corresponding \define{annihilation operator} is given by formal
differentiation:
\[
a_i \Psi = \frac{\partial}{\partial z_i} \Psi.
\]
This takes any pure state to the pure state with one
fewer instance of the $i$th species, but multiplied by a coefficient
$n_i$:
\[    a_i \left(z_1^{\ell_1} \cdots z_i^{\ell_i}\cdots
z_k^{\ell_k}\right) = \ell_i \; z_1^{\ell_1} \cdots z_i^{\ell_i-1} \cdots z_k^{\ell_k} .\] 
This represents the fact that there are $n_i$ ways to annihilate one of 
the instances of the $i$th species present.  

In what follows we use multi-index notation to define, for any $n \in
\N^k$,
\[   a^n = a_1^{n_1} \cdots a_k^{n_k} \]
and
\[   {a^\dagger}^n = {a_1^\dagger}^{n_1} \cdots {a_k^\dagger}^{n_k} .\]
These expressions are well-defined because all the annihilation operators
commute with each other, and similarly for the creation operators.

\begin{theorem} 
For any stochastic reaction network, the Hamiltonian is given by
\[
H = \sum_{\tau \in T} r(\tau) \, 
\left({a^\dagger}^{t(\tau)} - {a^\dagger}^{s(\tau)}\right) \, a^{s(\tau)} \;.
\]
\end{theorem}

\begin{proof} We start with a basic result on annihilation and creation
operators:
 
\begin{lemma} \label{lem.annihilate.create}
For any multi-indices $\ell, m$ we have
\[     {a^\dagger}^m z^\ell = z^{\ell + m} , \qquad
        a^m z^\ell = \ell^{\underline{m}} \, z^{\ell - m}. \]
\end{lemma} 

\begin{proof}  The first equation follows inductively from the
definition of the creation operators, which implies
\[    {a_i}^\dagger z^\ell 
= z_1^{\ell_1} \cdots z_i^{\ell_i + 1} \cdots z_k^{\ell_k}. \]
The second follows inductively from the definition of annihilation operators,
which implies
\[    a_i \, z^\ell = 
\ell_i \, z_1^{\ell_1} \cdots z_i^{\ell_i - 1} \cdots z_k^{\ell_k}. \]
Note that $\ell^{\underline{m}} = 0$ if $m_i > \ell_i$ for any $i$,
and in this case we also have $a^m z^\ell = 0$, so the equation $a^m
z^\ell = \ell^{\underline{m}} \, z^{\ell - m}$ holds in this case
because both sides vanish.
\end{proof}

We now prove the theorem.  For any multi-index $\ell$, Equation \ref{master1} gives
\[   \begin{array}{ccl}
H z^\ell &=&  \displaystyle{\sum_{\ell' \in \N^k} H_{\ell' \ell}  
\, z^{\ell'} } 
\\  \\
&=& \displaystyle{\sum_{\tau \in T}} \, r(\tau)  \, \ell^{{\underline{s(\tau)}}}  
\left( z^{\ell + t(\tau) - s(\tau)} - z^\ell \right) 
\end{array} \]
Using Lemma \ref{lem.annihilate.create} we obtain
\[   
H z^\ell =
\displaystyle{\sum_{\tau \in T}} \, r(\tau) \, \left({a^\dagger}^{t(\tau)} 
- {a^\dagger}^{s(\tau)} \right) \, a^{s(\tau)} \, z^\ell 
\]
and thus
\[   H = \displaystyle{\sum_{\tau \in T}} \, r(\tau) \, 
\left({a^\dagger}^{t(\tau)} - {a^\dagger}^{s(\tau)} \right) \, a^{s(\tau)}  .
\]
The last step requires a little justification.  Technically speaking, the 
monomials $z^\ell$ are not a basis of the stochastic Fock space
$\R[[z_1, \dots, z_k]]$, because not every formal power series is a 
\emph{finite} linear combination of monomials.  However, these
monomials form a `topological basis' of $\R[[z_1, \dots, z_k]]$, 
in the sense that every element of
this vector space can be expressed as an \emph{convergent infinite} 
linear combination of these monomials, with respect to a certain topology.
(In this topology, a sequence of formal power series converges if each
coefficient converges.)  The annihilation and creation operators, and
indeed all the operators discussed in this paper, are continuous in
this topology.  So, to check equations between these operators, it
suffices to check them on the monomials $z^\ell$.  \end{proof}

Next we investigate what the master equation says about the time evolution
of expected values.   For this, we start by defining
\[   \langle \Psi \rangle = \sum_{\ell \in \N^k} \psi_\ell \]
for any formal power series
\[    \Psi = \sum_{\ell \in \N^k} \psi_\ell z^\ell . \] 
In general the sum defining $\langle \Psi \rangle$ may not converge,
but it converges and equals 1 when $\Psi$ is a mixed state.  Suppose
$O$ is an operator on the space of formal power series in the
variables $z_1, \dots, z_k$.  We define the \define{expected value} of
$O$ in the mixed state $\Psi$ to be $\langle O \Psi \rangle$, assuming
this converges.  This generalizes the usual definition of the expected
value of a random variable, in a way that emphasizes the analogy
between stochastic mechanics and quantum mechanics.  While in quantum
mechanics we compute expected values of observables using expressions
of the form $\langle \Psi, O \Psi \rangle$, where the brackets denote
an inner product, in stochastic mechanics we use $\langle O \Psi
\rangle$.  For more details see \cite{{BaezBiamonte:2012},{BaezFong:2012}}.

We are especially interested in the expected value of the 
\define{number operators}
\[        N_i = a^\dagger_i a_i  . \]
Note that 
\[    \langle N_i \Psi \rangle = \sum_{\ell \in \N^k} \ell_i \psi_\ell \]
Since $\psi_\ell$ is the probability of being in a pure state where there
are $\ell_i$ things of the $i$th species, the above formula indeed gives the
expected value of the number of things of this species.

Our goal is to compute
\[    \displaystyle{ \frac{d}{dt} \langle N_i \Psi(t) \rangle } \]
assuming that $\Psi(t)$ obeys the master equation.  
For this, we need to define the falling powers of an operator $A$ as follows:
\[  A^{\underline{p}} = A(A-1) \cdots (A-p+1)  \]
for any $p \in \N$.   If $m$ is a multi-index we define
\[  N^{\underline{m}} 
= N_1^{\, \underline{m_1}} \cdots N_k^{\,\underline{m_k}}. \]
This is well defined because the number operators $N_i$ commute.

\begin{lemma} \label{lem.number.falling}
For any multi-indices $\ell$ and $m$ we have
\[  N^{\underline{m}} \, z^\ell = \ell^{\underline{m}} \, z^\ell .
\qedhere \]
\end{lemma}

\begin{proof}  By Lemma \ref{lem.annihilate.create} and the definition
 of the number operators we have 
\[   N_i z^\ell = \ell_i z^\ell  \]
for all $1 \le i \le k$.  The lemma follows directly from this.
\end{proof}

There is a way to extract a classical state from a mixed state, which lets
us connect the rate equation to the master equation.  To do this, we
 write $N$ for the vector, or list, of number
operators $(N_1, \dots, N_k)$.   Then, given a mixed state $\Psi$, we
set
\[    \displaystyle{ \langle N \Psi \rangle  =
 (  \langle N_1 \Psi  \rangle  , \; \dots \;, 
\langle N_k \Psi \rangle  ) } .\]
If the expected values here are well-defined, $\langle N \Psi \rangle 
\in [0,\infty)^k$ is a classical state.  It is a simplified description
of the mixed state $\Psi$, which only records the expected number of 
instances of each species.

The next result says how this sort of classical state evolves in time, 
assuming that the mixed state it comes from obeys the master equation:

\begin{theorem} \label{thm.expected.value}
For any stochastic reaction network and any mixed state
$\Psi(t)$ evolving in time according to the master equation, we have
\[   \displaystyle{ \frac{d}{dt} \langle N \Psi(t) \rangle } = 
 \displaystyle{\sum_{\tau \in T}} \, r(\tau) \,  (s(\tau) - t(\tau)) \; 
\left\langle N^{\underline{s(\tau)}}\, \Psi(t) \right\rangle  \]
assuming the expected values and their derivatives exist.
\end{theorem}

\begin{proof}  First note using Lemmas \ref{lem.annihilate.create} and 
\ref{lem.number.falling} that for any multi-indices $\ell, m, n$ and any 
natural number $i$ we have
\[ \begin{array}{ccl}
\langle N_i \, {a^\dagger}^{m} a^n z^\ell \rangle 
&=& \ell^{\underline{n}} \, \langle N_i \,{a^\dagger}^{m} z^{\ell - n} \rangle
\\  \\
&=& \ell^{\underline{n}} \, \langle N_i \,  z^{\ell - n + m} \rangle \\  \\
&=& \ell^{\underline{n}} \,(\ell_i - n_i + m_i)  .
\end{array}
\]
Thus if at any given time
\[  \Psi(t) = \sum_{\ell \in \N^k} \psi_\ell(t)  z^\ell  \]
and $\Psi(t)$ evolves in time according to the master equation, we have
\[  \begin{array}{ccl}
\displaystyle{ \frac{d}{dt} \langle N_i \Psi(t) \rangle } &=& 
\langle N_i H \Psi(t) \rangle 
\\   \\
&=& \displaystyle{\sum_{\tau \in T}} \, \displaystyle{\sum_{\ell \in \N^k}}\, 
r(\tau) \, \left\langle N_i \left({a^\dagger}^{t(\tau)} - 
{a^\dagger}^{s(\tau)}\right) \, a^{s(\tau)} \psi_\ell(t) z^\ell \right\rangle 
\\  \\
&=& \displaystyle{\sum_{\tau \in T}} \, \displaystyle{\sum_{\ell \in \N^k}}\, 
r(\tau) \, \ell^{\underline{s(\tau)}}\, (s_i(\tau) - t_i(\tau)) \, 
\psi_\ell(t)  
\\  \\
&=& \displaystyle{\sum_{\tau \in T}} \, 
\displaystyle{\sum_{\ell \in \N^k}}\, r(\tau) \,  (s_i(\tau) - t_i(\tau)) \; 
\left\langle N^{\underline{s(\tau)}}\, z^\ell \right\rangle \, \psi_\ell(t)   
\\   \\
&=& \displaystyle{\sum_{\tau \in T}} \, r(\tau) \,  (s_i(\tau) - t_i(\tau)) \;
 \left\langle N^{\underline{s(\tau)}}\, \Psi(t) \right\rangle   
\end{array} \]
Recall that $N$ is our notation for the vector of number operators
$(N_1, \dots, N_k)$.  Thus, the above equation is equivalent to
\[   \displaystyle{ \frac{d}{dt} \langle N \Psi(t) \rangle } = 
 \displaystyle{\sum_{\tau \in T}} \, r(\tau) \,  (s(\tau) - t(\tau)) \; 
 \left\langle N^{\underline{s(\tau)}}\, \Psi(t) \right\rangle   \qedhere
  \]
\end{proof}

The above theorem makes clear that the rate equation is closely
related to the master equation, since the former says
\[
\frac{d}{d t}x(t) = \sum_{\tau \in T} r(\tau) (t(\tau) - 
s(\tau)) x(t)^{s(\tau)}
\]
while the latter implies
\[   \displaystyle{ \frac{d}{dt} \langle N \Psi(t) \rangle } = 
 \displaystyle{\sum_{\tau \in T}} \, r(\tau) \,  (s(\tau) - t(\tau)) \; 
 \left\langle N^{\underline{s(\tau)}}\, \Psi(t) \right\rangle   .\]
Suppose we let $x(t)$ be the classical state coming from the mixed state
$\Psi(t)$:
\[     x(t) = \langle N \Psi(t) \rangle  \]
The rate equation for $x(t)$ would then follow from the master equation 
for $\Psi(t)$ if we had
\[   \langle N \Psi(t) \rangle^m = 
\langle N^{\underline{m}} \Psi(t) \rangle \] 
for every multi-index $m$.  This equation is not true in general.
However, it should hold \emph{approximately} in a suitable limit of
large numbers.  And surprisingly, it holds \emph{exactly} for coherent
states.

\section{Coherent states}

For any $c \in [0,\infty)^k$, we define the \define{coherent state} with
expected value $c$ to be the mixed state
\[
\Psi_c = \displaystyle{\frac{e^{c \cdot z}}{e^c}} 
\]
where $c \cdot z$ is the dot product of the vectors $c$ and $z$, and
we set $e^c = e^{c_1+\dots+c_n}$.  Equivalently,
\[  \Psi_c = \frac{1}{e^c} \sum_{n \in \N^k} \frac{c^n}{n!}z^n, \] 
where $c^n$ and $z^n$ are defined as products in our usual way, and
$n! = n_1! \, \cdots \, n_k!$. The name `coherent state' comes from
quantum mechanics \cite{KlauderSkagerstam:1980}, where we think of the coherent state
$\Psi_c$ as the quantum state that best approximates the classical
state $c$. In the state $\Psi_c$, the probability of having $n_i$
things of the $i$th species is equal to
\[      e^{-c_i} \, \displaystyle{\frac{ c_i^{n_i}}{n_i!} }. \] 
This is precisely the definition of a Poisson distribution with mean
equal to $c_i$. The state $\Psi_c$ is thus a product of independent
Poisson distributions.

\begin{theorem} \label{thm.coherent}
Given any stochastic reaction network, let $\Psi(t)$ be a mixed state
evolving in time according to the master equation.  If $\Psi(t)$ is a
coherent state when $t = t_0$, then
\[    x(t) = \langle N \Psi(t) \rangle \]
obeys the rate equation when $t = t_0$.
\end{theorem}

\begin{proof}  We prove this using a series of lemmas:

\begin{lemma} \label{lem.annihilation.coherent}
For any multi-index $m$ and any coherent state $\Psi_c$ we have
\[         a^m \Psi_c = c^m \Psi_c  \]
and
\[         {a^\dagger}^m \Psi_c = z^m \Psi_c  .\]
\end{lemma}

\begin{proof}  
The second equation is immediate from the definition, while the first
follows from
\[    a_i \Psi_c = 
\frac{\partial}{\partial z_i} \displaystyle{\frac{e^{c \cdot z}}{e^c}} = 
c_i \Psi_c . \qedhere \]
\end{proof}

\begin{lemma} \label{lem.number.falling.2}
For any multi-index $m$ we have
\[       N^{\underline{m}} = {a^{\dagger}}^m a^m . \]
\end{lemma}

\begin{proof} 
For any multi-index $\ell$, Lemmas \ref{lem.annihilate.create} and
\ref{lem.number.falling} imply
\[  
{a^\dagger}^m a^m z^\ell 
\;\; = \;\; \ell^{\underline{m}} \, {a^\dagger}^{m} z^{\ell - m}  
\;\; = \;\; \ell^{\underline{m}} \,  z^\ell  
\;\; = \;\; N^{\underline{m}} \,  z^\ell  .
\]
Since the states $z^\ell$ form a topological basis of $\R[[z_1, \dots,
z_k]]$, and all the operators in question are continuous, the lemma
follows.
\end{proof}

\begin{lemma} \label{lem.creation.stochastic}
For any multi-index $m$ and any $\Psi \in \R[[z_1, \dots, z_k]]$ we have
\[    \langle {a^{\dagger}}^m \Psi \rangle = \langle \Psi \rangle \]
\end{lemma}

\begin{proof}
We have
\[    \langle {a^{\dagger}}^m \Psi \rangle \;\; = \;\; 
\langle \sum_{\ell \in \N^k} \psi_\ell z^{\ell + m} \rangle \;\; = \;\; 
 \sum_{\ell \in \N^k} \psi_\ell  = \langle \Psi \rangle .\qedhere \]
\end{proof}

\begin{lemma} \label{lem.coherent}
For any multi-index $m$ and any coherent state $\Psi_c$ we have
\[    \langle N^{\underline{m}} \Psi_c \rangle 
= \langle N \Psi_c \rangle^m  \]
\end{lemma}

\begin{proof} 
By Lemma \ref{lem.number.falling.2} and Lemma
\ref{lem.annihilation.coherent} we have
 \[ N^{\underline{m}} \Psi_c =  {a^\dagger}^m a^m \Psi_c \;\; = \;\; 
c^m {a^\dagger}^m \Psi_c \] so by Lemma \ref{lem.creation.stochastic}
\[    \langle N^{\underline{m}} \Psi_c \rangle \; \; = \; \; 
\langle c^m {a^\dagger}^m \Psi_c \rangle  \; \; 
= \; \;  c^m \langle \Psi_c \rangle  \; \; = \; \;  c^m.  \]
As a special case we have $\langle N_i \Psi_c \rangle = c_i$, so we also have
\[   \langle N \Psi_c \rangle^m = c^m  \]
and it follows that $\langle N^{\underline{m}} \Psi_c \rangle = \langle N
\Psi \rangle^m$.
\end{proof}

To prove Theorem \ref{thm.coherent}, we start by setting
\[      x(t) = \langle N \Psi(t) \rangle \]
where $\Psi(t)$ is a solution of the master equation.
Theorem \ref{thm.expected.value} implies that the time derivative of
$x$ equals
\[   \dot{x}(t) = 
 \displaystyle{\sum_{\tau \in T}} \, r(\tau) \, (s(\tau) - t(\tau)) \;
  \left\langle N^{\underline{s(\tau)}}\, \Psi(t) \right\rangle \;. \] 
If  $\Psi(t)$ is a coherent state at time $t = t_0$, Lemma
\ref{lem.coherent} implies that
\[  \left\langle N^{\underline{s(\tau)}}\, \Psi(t_0) \right\rangle =
\left\langle N \Psi(t_0) \right\rangle^{s(\tau)}  = x(t_0)^{s(\tau)}   \]
so we have
\[     \dot{x}(t_0) = 
\displaystyle{\sum_{\tau \in T}} \, r(\tau) \, (s(\tau) - t(\tau)) \; 
x(t_0)^{s(\tau)}   \]
which is the rate equation at time $t_0$, as desired.
\end{proof}

In general, the above result applies at only one moment in
time.  The reason is that if a solution of the master equation
is a coherent state at some time, it need not be a coherent
state at later (or earlier) times.  Still, Theorem
\ref{thm.expected.value} implies that as long as the expected values
$\langle N^{\underline{m}} \Psi(t) \rangle$ are \emph{close} to the powers 
$\langle N \Psi(t) \rangle^m$, at least when $m$ is the source of some
reaction, the expected values $x(t) = \langle N \Psi(t) \rangle$ will
\emph{approximately} obey the rate equation. This could be studied in
detail using either the techniques introduced here or those discussed
by Anderson and Kurtz \cite{AndersonKurtz:2011}.

Furthermore, in some cases when $\Psi(t)$ is initially a
coherent state it continues to be so for all times.  In this case,
$x(t)$ continues to exactly obey the rate equation.  For example, this
is true when all the complexes in the reaction network consist of a
single species.  It is also true for a large class of equilibrium
solutions of the master equation, as shown by Anderson, Craciun and
Kurtz \cite{AndersonCraciunKurtz:2010}.  For more on how quantum 
techniques apply to this situation, see the companion to this paper \cite{BaezFong:2014}.

\subsection*{Acknowledgements}
I thank Jacob Biamonte, Brendan Fong, and the students at U.\ C.\
Riverside who helped come up with a preliminary version of the proof of 
Theorem \ref{thm.expected.value}, notably Daniel Estrada, Reeve Garrett, 
Michael Knap, Tu Pham, Blake Pollard and Franciscus Rebro.  I thank
John Rowlands and Blake Stacey for catching typos and other mistakes, and
the anonymous referees for suggesting improvements.  
I also thank the Centre of Quantum
Technology and the Mathematics Department of U.\ C.\ Riverside, where
this work was done.

\end{document}